\documentclass[11pt]{article}

\usepackage{odonnell}

\newcommand{\CStar}[1]{c^*_{#1}}
\newcommand{\SStar}[1]{s^*_{#1}}
\newcommand{\pgeq}{\succeq}
\newcommand{\pleq}{\preceq}

\newcommand{\lift}{\mathfrak{L}}

\newcommand{\bbi}{\mathds{1}} 
\newcommand{\OPT}{\textsc{Opt}} \newcommand{\Opt}{\OPT}
\newcommand{\Eig}{\textsc{Eig}}

\newcommand{\GW}{\textsc{Sdp}}
\newcommand{\SDP}{\GW}
\newcommand{\GWDual}{\textsc{Sdp-Dual}}
\newcommand{\GenEig}{\textsc{PartSdp}}
\newcommand{\GenEigDual}{\textsc{PartSdp-Dual}}

\newcommand{\geneig}{partitioned SDP\xspace}
\newcommand{\Geneig}{Partitioned SDP\xspace}

\newcommand{\oldu}{\zeta}

\usepackage{tikz}

\begin{document}

\title{The SDP value of random 2CSPs}
\author{Amulya Musipatla%
\thanks{Computer Science Department, Carnegie Mellon University.}
\and
Ryan O'Donnell\footnotemark[1]\;%
\thanks{Supported by NSF grant FET-1909310.
    This material is based upon work
    supported by the National Science Foundation under grant numbers
    listed above. Any opinions, findings and conclusions or
    recommendations expressed in this material are those of the author
    and do not necessarily reflect the views of the National Science
    Foundation (NSF).}
\and
Tselil Schramm%
\thanks{Department of Statistics, Stanford University.}
\and
Xinyu Wu\footnotemark[1]\;\footnotemark[2]
}

\date{\today}

\maketitle

\vspace{-.2in}

\begin{abstract}
    We consider a very wide class of models for sparse random Boolean 2CSPs; equivalently, degree-2 optimization problems over~$\{\pm 1\}^n$.
    For each model $\calM$, we identify the ``high-probability value''~$\SStar{\calM}$ of the natural SDP relaxation (equivalently, the quantum value).
    That is, for all $\eps > 0$ we show that the SDP optimum of a random $n$-variable instance is (when normalized by~$n$) in the range $(\SStar{\calM}-\eps, \SStar{\calM}+\eps)$ with high probability.
    Our class of models includes non-regular CSPs, and ones where the SDP relaxation value is strictly smaller than the spectral relaxation value.
\end{abstract}

\section{Introduction} \label{sec:intro}

A large number of important algorithmic tasks can be construed as constraint satisfaction problems (CSPs): finding an assignment to Boolean variables to optimize the number of satisfied constraints.
Almost every form of constraint optimization is $\mathsf{NP}$-complete; thus one is led to questions of efficiently finding near-optimal solutions, or understanding the complexity of average-case rather than worst-case instances.
Indeed, understanding the complexity of random sparse CSPs is of major importance not just in traditional algorithms theory, but also in, e.g., cryptography~\cite{JLS20}, statistical physics~\cite{MM09}, and learning theory~\cite{DSS14}.

Suppose we fix the model $\calM$ for a random sparse CSP on~$n$ variables (e.g., random $k$-SAT with a certain clause density).
Then it is widely believed that there should be a constant~$\CStar{\calM}$ such that the optimal value of a random instance is $\CStar{\calM} \pm o_{n\to\infty}(1)$ with high probability (whp).
(Here we define the optimal value to mean the maximum number of simultaneously satisfiable constraints, divided by the number of variables.)
Unfortunately, it is extremely difficult to prove this sort of result; indeed, it was considered a major breakthrough when Bayati, Gamarnik, and Tetali~\cite{BGT13} established it for one of the simplest possible cases: Max-Cut on random $d$-regular graphs (which we will denote by $\mathcal{MC}_d$).
Actually ``identifying'' the value~$\CStar{\calM}$ (beyond just proving its existence) is even more challenging.
It is generally possible to estimate $\CStar{\calM}$ using heuristic methods from statistical physics, but making these estimates rigorous is beyond the reach of current methods.
Taking again the example of Max-Cut on $d$-regular random graphs, it was only recently~\cite{DMS17} that the value $\CStar{\mathcal{MC}_d}$ was determined up to a factor of $1 \pm o_{d \to \infty}(1)$.
The value for any particular~$d$, e.g.\ $\CStar{\mathcal{MC}_3}$, has yet to be established.

Returning to algorithmic questions, we can ask about the \emph{computational feasibility} of optimally solving sparse random CSPs.
There are two complementary questions to ask: given a random instance from model~$\calM$ (with presumed optimal value~$\CStar{\calM} \pm o_{n\to \infty}(1)$), can one efficiently \emph{find} a solution achieving value \mbox{$\gtrapprox \CStar{\calM}$}, and can one efficiently \emph{certify} that every solution achieves value \mbox{$\lessapprox~\CStar{\calM}$}?
The former question is seemingly a bit more tractable; for example, a very recent breakthrough of Montanari~\cite{Mon21} gives an efficient algorithm for (whp) finding a cut in a random graph $\calG(n,p)$ graph of value at least $(1-\eps) \CStar{\calG(n,p)}$.
On the other hand, we do not know any algorithm for efficiently certifying (whp) that a random instance has value at most~$(1+\eps) \CStar{\calG(n,p)}$.
Indeed, it reasonable to conjecture that no such algorithm exists, leading to an example of a so-called ``information-computation gap''.

To bring evidence for this we can consider \emph{semidefinite programming} (SDP), which provides efficient algorithms for certifying an upper bound on the optimal value of a CSP~\cite{FL92}.
Indeed, it is known~\cite{Rag09} that, under the Unique Games Conjecture, the basic SDP relaxation provides essentially optimal certificates for CSPs in the \emph{worst case}.
In this paper we in particular consider Boolean 2CSPs --- more generally, optimizing a homogeneous degree-$2$ polynomial over the hypercube --- as this is the setting where semidefinite programming is most natural.
Again, for a fixed model $\calM$ of random sparse Boolean 2CSPs, one expects there should exist a constant~$\SStar{\calM}$ such that the optimal SDP-value of an instance from~$\calM$ is whp $\SStar{\calM} \pm o_{n\to\infty}(1)$.
Philosophically, since semidefinite programming is doable in polynomial time, one may be more optimistic about proving this and explicitly identifying~$\SStar{\calM}$.
Indeed, some results in this direction have recently been established.

\subsection{Prior work on identifying high-probability SDP values}
Let us consider the most basic case: $\mathcal{MC}_d$, Max-Cut on random $d$-regular graphs.
For ease of notation, we will consider the equivalent problem of maximizing $\tfrac{1}{n}x^\transp (-\bA) x$ over $x \in \{\pm 1\}^n$, where $\bA$ is the adjacency matrix of a random $n$-vertex $d$-regular graph.\footnote{Throughout this work, \textbf{boldface} is used to denote random variables.}
Although $\SStar{\mathcal{MC}_d}$, the high-probability SDP relaxation value, was pursued as early as 1987~\cite{Bop87} (see also~\cite{DH73}), it was not until 2015 that Montanari and Sen~\cite{MS16} established the precise result $\SStar{\mathcal{MC}_d} = 2 \sqrt{d-1}$.  That is, in a random $d$-regular graph, whp the basic SDP relaxation value~\cite{Bop87,RW95,DP93,PR95} for the size of the maximum cut is $(\frac{d}{4} + \sqrt{d-1} \pm o_{n\to\infty}(1)) n$.
Here the special number $2\sqrt{d-1}$ is the maximum eigenvalue of the $d$-regular infinite tree.

The proof of this result has two components: showing $\GW(-\bA) \geq 2 \sqrt{d-1} - \eps$ whp, and showing $\GWDual(-\bA) \leq 2 \sqrt{d-1} + \eps$ whp.
Here $\GW(A) = \max \{\la \rho, A \ra : \rho \pgeq 0,\;\rho_{ii} = \tfrac1n\;\forall i\}$ denotes the ``primal'' SDP value on matrix~$A$ (commonly associated with the Goemans--Williamson rounding algorithm~\cite{GW95}), and $\GWDual(A) = \min \{\lambda_{\max}(A + \diag(\oldu)) : \avg_i(\oldu_i) = 0\}$  denotes the (equal) ``dual'' SDP value on~$A$.
To show the latter bound, it is sufficient to observe that $\GWDual(-\bA) \leq \lambda_{\max}(-\bA)$, the ``eigenvalue bound'', and  $\lambda_{\max}(-\bA) \leq 2\sqrt{d-1} + o_n(1)$ whp by Friedman's Theorem~\cite{Fri08}.
As for lower-bounding $\GW(-\bA)$, Montanari and Sen used the ``Gaussian wave'' method~\cite{Elo09,CGHV15,HV15} to construct primal SDP solutions achieving at least $2 \sqrt{d-1} - \eps$ (whp).
The idea here is essentially to build the SDP solutions using an approximate eigenvector (of finite support) of the infinite $d$-regular tree achieving eigenvalue $2 \sqrt{d-1} - \eps$; the fact that SDP constraint ``$\rho_{ii} = \tfrac1n \;\forall i$'' can be satisfied relies heavily on the regularity of the graph.

\begin{remark}  \label{rem:theme}
    The Montanari--Sen result in passing establishes that the (high-probability) eigenvalue and SDP bounds coincide for random regular graphs.
    This is consistent with a known theme, that the two bounds tend to be the same (or nearly so) for graphs where ``every vertex looks similar'' (in particular, for regular graphs).
    This theme dates back to Delorme and Poljak~\cite{DP93}, who showed that $\GWDual(-A) = \lambda_{\max}(-A)$ whenever $A$ is the adjacency matrix of a vertex-transitive graph.
\end{remark}

Subsequently, the high-probability SDP value $\SStar{\calM}$ was established for a few other models of random regular 2CSPs.
Deshpande, Montanari, O'Donnell, Schramm, and Sen~\cite{DMOSS19} showed that for $\calM = \mathcal{N\mkern-5mu AE}3_c$ --- meaning random regular instances of NAE-3SAT (not-all-equals 3Sat) with each variable participating in~$c$ clauses --- we have $\SStar{\calM} = \frac{9}{8} - \frac38 \cdot \frac{\sqrt{c-1}- \sqrt{2}}{c}$.
We remark that NAE-3SAT is effectively a 2CSP, as the predicate $\mathrm{NAE}_3 : \{\pm 1\}^3 \to \{0,1\}$ may be expressed as $\frac34 - \frac14(xy + yz + zx)$, supported on the ``triangle'' formed by variables $x,y,z$.
The analysis in this paper is somewhat similar to that in~\cite{MS16}, but with the infinite graph $\mathfrak{X} = K_3 \star K_3 \star \cdots \star K_3$ ($c$~times) replacing the $d$-regular infinite tree.
This $\mathfrak{X}$ is the $2c$-regular infinite ``tree of triangles'' depicted (partly, in the case $c=3$) in \Cref{fig:c3c3c3}.
More generally, \cite{DMOSS19}~established the high-probability SDP value for large random (edge-signed) graphs that locally resemble $K_r \star K_r \star \cdots \star K_r$, the $(r-1)c$-regular infinite ``tree of cliques~$K_r$''.
(The $r = 2$ case essentially generalizes~\cite{MS16}.)
As in \cite{MS16}, $\SStar{\calM}$ coincides with the (high-probability) eigenvalue bound.
The upper bound on $\SStar{\calM}$ is shown by using Bordenave's proof~\cite{Bor20} of Friedman's Theorem for random $(c,r)$-biregular graphs.
The lower bound on $\SStar{\calM}$ is shown using the Gaussian wave technique, relying on the distance-regularity of the graphs $K_r \star K_r \star \cdots \star K_r$ (indeed, it is known that every infinite distance-regular graph is of this form).

\myfig{.255}{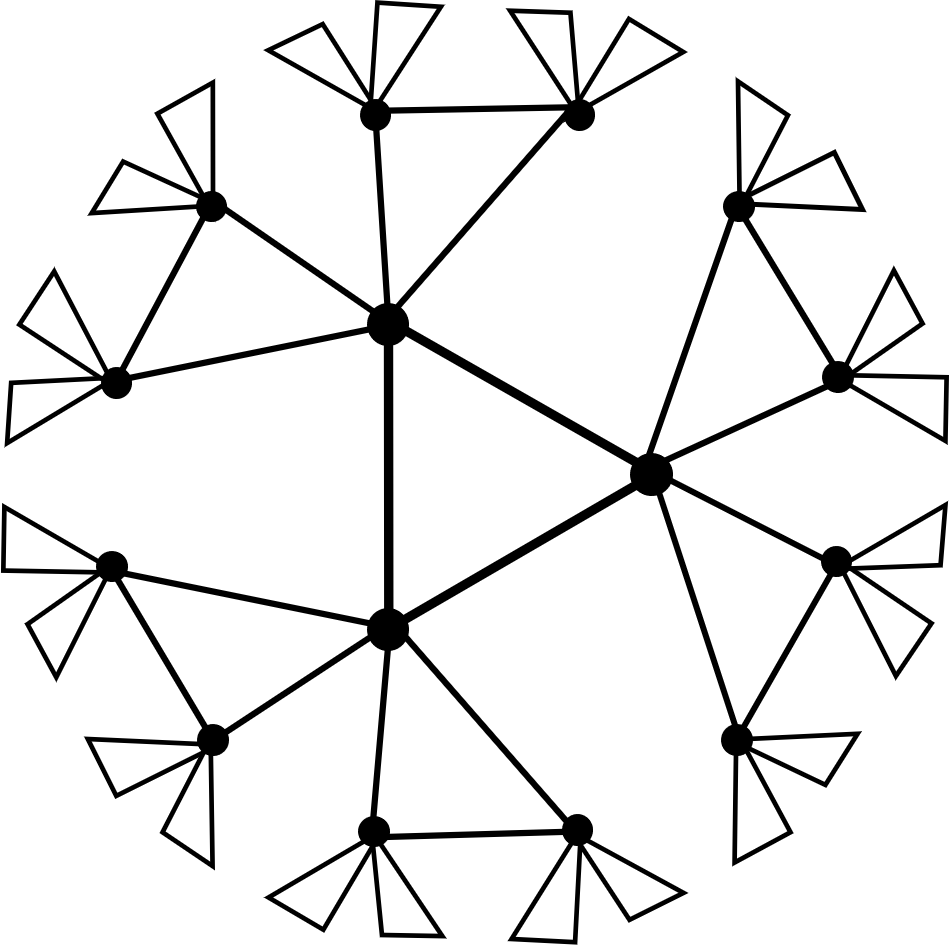}{The $6$-regular infinite graph $K_3 \star K_3 \star K_3$, modeling random $3$-regular NAE3-SAT.}{fig:c3c3c3}

Mohanty, O'Donnell, and Paredes~\cite{MOP20} generalized the preceding two results to the case of ``two-eigenvalue'' 2CSPs.
Roughly speaking, these are 2CSPs formed by placing copies of a small weighted ``constraint graph''~$\calH$ --- required to have just two distinct eigenvalues --- in a random regular fashion onto~$n$ vertices/variables.
(This is indeed a generalization~\cite{DMOSS19}, as cliques have just two distinct eigenvalues.)
As two-eigenvalue examples, \cite{MOP20}~considered CSPs with the ``CHSH constraint'' --- and its generalizations, the ``Forrelation${}_k$'' constraints --- which are important in quantum information theory~\cite{CHSH69,AA18}.
Here the SDP value of an instance is particularly relevant as it is precisely the optimal ``quantum entangled value'' of the 2CSP~\cite{CHTW04}.
Once again, it is shown in~\cite{MOP20} that the high-probability SDP and eigenvalue bounds coincide for these types of CSPs.
The two-eigenvalue condition is used at a technical level in both the variant of Bordenave's theorem proven for the eigenvalue upper bound, and in the Gaussian wave construction in the SDP lower bound.

Most recently, O'Donnell and Wu~\cite{OW20} used the results of Bordenave and Collins~\cite{BC19} (a substantial generalization of~\cite{Bor20}) to establish the high-probability \emph{eigenvalue relaxation bound} for  very wide range of random 2CSPs, encompassing all those previously mentioned: namely, any quadratic optimization problem defined by random ``matrix polynomial lifts'' over literals.

\subsection{Our work}

In this work, we establish the high-probability SDP value $\SStar{\calM}$ for random instances of any 2CSP model~$\calM$ arising from lifting matrix polynomials (as in~\cite{OW20}).
This generalizes all previously described work on SDP values, and covers many more cases, including random lifts of any base 2CSP and random graphs modeled on any free/additive/amalgamated product.
Such graphs have seen numerous applications within theoretical computer science, for example the zig-zag product in derandomization (e.g.~\cite{RVW02}) and lifts of 2CSPs in the study of the stochastic block model (e.g~\cite{BBKMW20}).
See \Cref{sec:prelims} for more details and definitions, and see~\cite{OW20} for a thorough description of the kinds of random graph/2CSP models that can arise from matrix polynomials.

Very briefly, a matrix polynomial~$p$ is a small, explicit ``recipe'' for producing random $n$-vertex edge-weighted graphs, each of which ``locally resembles'' an associated \emph{infinite} graph~$\mathfrak{X}_p$.
For example, $p_3(Y_1,Y_2,Y_3) \coloneqq Y_1 + Y_2 + Y_3$ is a recipe for random (edge-signed) $3$-regular $n$-vertex graphs, and here $\mathfrak{X}_{p_3}$ is the infinite $3$-regular tree.
As another example, if $p_{333}(Z_{1,1}, \dots, Z_{3,3})$ denotes the following matrix polynomial ---
\begin{equation} \label{eqn:p333}
    \begin{pmatrix}
                  0
                & Z_{1,1} Z_{1,2}^* + Z_{2,1} Z_{2,2}^* + Z_{3,1} Z_{3,2}^*
                & Z_{1,1} Z_{1,3}^* + Z_{2,1} Z_{2,3}^* + Z_{3,1} Z_{3,3}^*  \\
                   Z_{1,2} Z_{1,1}^* + Z_{2,2} Z_{2,1}^* + Z_{3,2} Z_{3,1}^*
                & 0
                & Z_{1,2} Z_{1,3}^* + Z_{2,2} Z_{2,3}^* + Z_{3,2} Z_{3,3}^* \\
                    Z_{1,3}Z_{1,1}^* + Z_{2,3}Z_{2,1}^*  + Z_{3,3}Z_{3,1}^*
                &  Z_{1,3}Z_{1,2}^* +  Z_{2,3}Z_{2,2}^* + Z_{3,3}Z_{3,2}^*
                & 0
          \end{pmatrix}
\end{equation}
--- then $p_{333}$ is a recipe for random (edge-signed) $6$-regular $n$-vertex graphs where every vertex participates in $3$~triangles.
In this case, $\mathfrak{X}_{p_{333}}$ is the infinite graph (partly) depicted in \Cref{fig:c3c3c3}.
The Bordenave--Collins theorem~\cite{BC19} shows that if $\bA$ is the adjacency matrix of a random \emph{unsigned} $n$-vertex graph produced from a matrix polynomial~$p$, then whp the ``nontrivial'' spectrum of~$\bA$ will be within~$\eps$ (in Hausdorff distance) of the spectrum of~$\mathfrak{X}_p$.
In the course of derandomizing this theorem, O'Donnell and Wu~\cite{OW20} established that for random \emph{edge-signed} graphs, the modifier ``nontrivial'' should be dropped.
As a consequence, in the signed case one gets $\lambda_{\max}(\bA) \approx \lambda_{\max}(\mathfrak{X}_p)$ up to an additive~$\eps$, whp; i.e., the high-probability eigenvalue bound for CSPs derived from~$p$ is precisely~$\lambda_{\max}(\mathfrak{X}_p)$.
We remark that for simple enough~$p$ there are formulas for $\lambda_{\max}(\mathfrak{X}_p)$; regarding our two example above, it is $2\sqrt{2}$ for $p = p_3$, and it is~$5$ for $p = p_{333}$.
In particular, if $p$ is a linear matrix polynomial, $\lambda_{\max}(\mathfrak{X}_p)$ may be determined numerically with the assistance of a formula of Lehner~\cite{Leh99} (see also~\cite{GK21} for the case of standard random lifts of a fixed base graph).\\

In this paper we investigate the high-probability \emph{SDP} value --- denote it $s^*_p$ --- of a large random 2CSP (Boolean quadratic optimization problem) produced by a matrix polynomial~$p$.
Critically, our level of generality lets us consider \emph{non-regular} random graph models, in contrast to all previous work.
Because of this, see we cases in contrast to \Cref{rem:theme}, where (whp) the SDP value is strictly smaller than the eigenvalue relaxation bound.
As a simple example, for random edge-signed $(2,3)$-biregular graphs, the high-probability eigenvalue bound is $\sqrt{2-1} + \sqrt{3-1} = 1+ \sqrt{2} \approx 2.414$, but our work establishes that the high-probability SDP value is $\sqrt{\frac{13}{4} + 2\sqrt{2}} - \frac{1}{10} \approx 2.365$.\\

An essential part of our work is establishing the appropriate notion of the ``SDP value'' of an infinite graph~$\mathfrak{X}_p$, with adjacency operator~$A_\infty$.
While the eigenvalue bound $\lambda_{\max}(A_\infty)$ makes sense for the infinite-dimensional operator~$A_\infty$, the SDP relaxation does not.
The definition $\GW(A_\infty) = \max \{\la \rho, A_\infty \ra : \rho \pgeq 0,\;\rho_{ii} =\tfrac1n\;\forall i\}$ does not make sense, since ``$n$'' is~$\infty$.
Alternatively, if one tries the normalization $\rho_{ii} = 1$, any such~$\rho$ will have infinite trace, and hence $\la \rho, A_\infty \ra$ may be~$\infty$.
Indeed, since the only control we have on~$A_\infty$'s ``size'' will be an operator norm (``$\infty$-norm'') bound, the expression $\la \rho, A_\infty \ra$ is only guaranteed to make sense if~$\rho$ is restricted to be trace-class (i.e., have finite ``$1$-norm'').

On the other hand, we know that the eigenvalue bound $\lambda_{\max}(A_\infty)$ is too weak, intuitively because it does not properly ``rebalance'' graphs~$\mathfrak{X}_p$ that are not regular/vertex-transitive.
The key to obtaining the correct bound is introducing a new notion, intermediate between the eigenvalue and SDP bounds, that is appropriate for graphs $\mathfrak{X}_p$ arising from matrix polynomial recipes.
Although these graphs may be irregular, their definition also allows them to be viewed as \emph{vertex-transitive} infinite graphs with $r \times r$ \emph{matrix} edge-weights.
In light of their vertex-transitivity, \Cref{rem:theme} suggests that a ``maximum eigenvalue''-type quantity --- suitably defined for matrix-edge-weighted graphs --- might serve as the sharp replacement for SDP value.
We introduce such a quantity, calling it the \emph{\geneig} bound.
Let $G$ be an $n$-vertex graph with $r \times r$ matrices as edge-weights, and let $A$ be its adjacency matrix, thought of as a Hermitian $n \times n$ matrix whose entries are $r \times r$ matrices.
We will define
\begin{equation} \label{eqn:geneig}
    \GenEig(A) = \sup \{\la \rho, A \ra : \rho \pgeq 0,\;\tr(\rho)_{ii} = \tfrac1r\},
\end{equation}
where here $\tr(\rho)$ refers to the $r \times r$ matrix obtained by summing the entries on $A$'s main diagonal (themselves $r \times r$ matrices), and $\tr(\rho)_{ii}$ denotes the scalar in the $(i,i)$-position of~$\tr(\rho)$.
This \geneig bound may indeed be regarded as intermediate between the maximum eigenvalue and the SDP value.
On one hand, given an scalar-edge-weighted $n$-vertex graph with adjacency matrix~$A$, we may take~$r = 1$ and then it is easily seen that $\GenEig(A)$ coincides with $\lambda_{\max}(A)$.
On the other hand, if we regard $A$ as a $1 \times 1$ matrix and take $r = n$ (so that we a have single vertex with a self-loop weighted by all of~$A$), then $\GenEig(A) = \SDP(A)$.

As we will see in \Cref{sec:sdp-duality}, $\GenEig(A)$ can be suitably defined even for bounded-degree infinite graphs with $r \times r$ edge-weights.
Furthermore, it has the following SDP dual:
\[
    \GenEigDual(A) = \inf \{\lambda_{\max}(A + \bbone_{n \times n} \otimes \diag(\oldu)) : \avg(\oldu_1, \dots, \oldu_r) = 0\}.
\]
In the technical \Cref{sec:dual2}, we show that there is no SDP duality gap between $\GenEig(A)$ and $\GenEigDual(A)$, even in the case of infinite graphs.
It is precisely the common value of $\GenEig(\mathfrak{X}_p)$ and $\GenEigDual(\mathfrak{X}_p)$ that is the high-probability SDP value of large random 2CSPs produced from~$p$; our main theorem is the following:
\begin{theorem} \label{thm:ourmain}
    Let $p$ be a matrix polynomial with $r \times r$ coefficients.
    Let $A_\infty$ be the adjacency operator (with $r \times r$ entries) of the associated infinite lift~$\mathfrak{X}_p$, and write $s^*_p = \GenEig(A_\infty) = \GenEigDual(A_\infty)$.
    Then for any $\eps,\beta > 0$ and sufficiently large~$n$, if $\bA_n$ is the adjacency matrix of a random edge-signed $n$-lift of~$p$, it holds that $s^*_p - \eps \leq \SDP(\bA_n) \leq s^*_p + \eps$ except with probability at most~$\beta$.
\end{theorem}

\noindent Note that $\GenEig(A_\infty)$ is a fixed value only dependent on the polynomial~$p$, a finitary object.\\

The upper bound $\SDP(\bA_n) \leq \GenEigDual(A_\infty) + \eps$ in this theorem can be derived from the results of~\cite{BC19,OW20}.
Our main work is to prove the lower bound $\SDP(\bA_n) \geq \GenEig(A_\infty) - \eps$.
For this, our approach is inspired by the Gaussian Wave construction of \cite{MS16,DMOSS19} for $d$-regular graphs (in the random lifts model), which can be viewed as constructing a feasible $\SDP(\bA_n)$ solution of value $\lambda_{\max}(A_\infty) - \eps$ using a truncated eigenfunction of $A_\infty$.
Since local neighborhoods in $A_\infty$ look like local neighborhoods in $\bA_n$ with high probability, the eigenfunction can be ``pasted'' almost everywhere into the graph $\bG_n$, which gives an SDP solution of value near $\lambda_{\max}(A_\infty)$.

This approach runs into clear obstacles in our setting.
Indeed, the raw eigenfunctions are of no use to us, as {\em the SDP value may be smaller than the spectral relaxation value}.
Instead, we first show that there is a $\rho_0$ with only finitely many nonzero entries that achieves the $\sup$ in \Cref{eqn:geneig} up to~$\eps$.
This is effectively a finite $r \times r$ matrix-edge-weighted graph.
We then show that this~$\rho_0$ can (just as in the regular case) whp be ``pasted'' almost everywhere into the graph $\bG_n$ defined by~$\bA_n$, which gives an SDP solution of value close to $\GenEig(A_\infty)$.
The fact that $\mathfrak{X}_p$ and $\bG_n$ are regarded as regular tree-like graphs with matrix edge-weights (rather than as irregular graphs with scalar edges-weights) is crucially used to show that the ``pasted solution'' satisfies the finite SDP's constraints ``$\rho_{ii} = \frac1n\; \forall i$''.

\section{Preliminaries} \label{sec:prelims}

To preface the following definitions and concepts, we remark that our ``real'' interest is in graphs (or 2CSPs) with real scalar weights.
The introduction of matrix edge-weights facilitates two things: helping us define a wide variety of interesting scalar-weighted graphs via matrix polynomial lifts; and, facilitating the definition of $\GenEig(\cdot)$, which we use to bound the SDP relaxation value of the associated 2CSPs.
Our use of potentially complex matrices is also not really essential; we kept them because  prior work that we utilize (\cite{BC19}, the tools in \Cref{sec:sdp-duality}) is stated in terms of complex matrices.
However the reader will not lose anything by assuming that all Hermitian matrices are in fact symmetric real matrices.

\subsection{Matrix-weighted graphs}

In the following definitions, we'll restrict attention to graphs with at-most-countable vertex sets~$V$ and bounded degree.
We also often use bra-ket notation, with $(\ket{v})_{v \in V}$ denoting the standard orthonormal basis for the complex vector space~$\ell_2(V)$.

\begin{definition}[Matrix-weighted graph]
    Fix any $r \in \N^+$.
    A \emph{matrix-weighted graph} will refer to a directed simple graph $G = (V,E)$ with self-loops allowed, in which each directed edge~$e$ has an associated \emph{weight} $a_e \in \C^{r \times r}$.
    If $(v,w) \in E \implies (w,v) \in E$ and $a_{(w,v)} = a_{(v,w)}^*$, we say that $G$ is an \emph{undirected} matrix-weighted graph.
    The \emph{adjacency matrix} of $G$ is the operator~$A$, acting on~$\ell_2(V) \otimes \C^r$, given by
    \[
        \sum_{(v,w) \in E} \ketbra{w}{v} \otimes a_{(v,w)}.
    \]
    It can be helpful to think to think of $A$ in matrix form, as a $|V| \times |V|$ matrix whose entries are themselves $r \times r$ edge-weight matrices.
    Note that if $G$ is undirected if and only if $A$ is self-adjoint, $A = A^*$.
\end{definition}

\begin{definition}[Extension of a matrix-weighted adjacency matrix]    \label{def:extension}
	Given a $|V| \times |V|$ matrix $A$ with $r \times r$ entries, we may also view it as a $|V|r \times |V|r$ matrix with scalar entries.
    When we wish to explicitly call attention to the distinction, we will call the latter matrix the \emph{extension} of $A$, and denote it by~$\widetilde{A}$.
\end{definition}

\subsection{Matrix polynomials}
\begin{definition}[Matrix polynomial]
	Let $Y_1, \ldots, Y_d$ be formal indeterminates that are their own inverses, and let $Z_1, \ldots, Z_e$ be formal indeterminates with formal inverses $Z_1^*, \ldots, Z_e^*$.
	For a fixed $r$, we define a \emph{matrix polynomial} to be a formal noncommutative polynomial~$p$ over the indeterminates $Y_1, \ldots , Z^*_e$, with coefficients in $\C^{r \times r}$.
    In particular, we may write
    \[
        p = \sum_{w} a_w w,
    \]
    where the sum is over words~$w$ on the alphabet of indeterminates, each $a_w$ is in $\C^{r \times r}$, and only finitely many~$a_w$ are nonzero.
    Here we call a word \emph{reduced} if it has no adjacent $Y_i Y_i$ or $Z_i Z_i^*$ pairs.
    We will denote the empty word by~$\bbi$.
\end{definition}

As we will shortly describe, we will always be considering substituting unitary operators for the $Z_i$'s, and unitary involution operators for the $Y_i$'s.
Thus we can think of $Z_i^*$ as both the inverse \emph{and} the ``adjoint'' of indeterminate~$Z_i$, and similarly we think of $Y_i^* = Y_i$.

\begin{definition}[Adjoint of a polynomial]
	Given a matrix polynomial $p = \sum_w a_w w$ as above, we define its \emph{adjoint} to be
\[
    p^* = \sum_w a_w^* w^*,
\]
    where $a^*$ is the usual adjoint of $a \in \C^{r \times r}$, and $w^*$ is the adjointed reverse of~$w$.
    That is, if $w = w_1 \cdots w_k$ then $w^* = w_k^* \cdots w_1^*$,  where $\bbi^* = \bbi$,  $Y_i^* = Y_i$, and $Z_i^{**} = Z_i$.
    We say $p$ is \emph{self-adjoint} if $p^* = p$ formally.
\end{definition}

Note that in any self-adjoint polynomial, some terms will be self-adjoint, and others will come in self-adjoint pairs.
In this work, we will only be considering self-adjoint polynomials.

\subsection{Lifts of matrix polynomials} \label{sec:lifts}

\begin{definition}[$n$-lift]    \label{def:n-lift}
    Given a matrix polynomial over the indeterminates $Y_1, \ldots, Z^*_e$, we define an \emph{$n$-lift} to be a sequence $\calL = (M_1, \dots, M_d, P_1, \dots, P_e)$ of $n \times n$ matrices, where each~$P_i$ is a signed permutation matrix and each $M_i$ is a signed matching matrix.\footnote{A signed matching matrix is the adjacency matrix of a perfect matching with $\pm 1$ edge-signs. If $d > 0$ then we must restrict to even~$n$.}
    A \emph{random $n$-lift} $\bcalL = (\bM_1, \dots, \bM_d, \bP_1, \dots, \bP_e)$ is one where the matrices are chosen independently and uniformly at random.
\end{definition}
\begin{definition}[Evaluation/substitution of lifts.] \label{def:eval}
    Given an $n$-lift~$\calL$ and a word~$w$, we define $\calL^w$ to be the $n \times n$ operator obtained by substituting appropriately into~$w$: namely, $Y_i = M_i$, $Z_i = P_i$, and $Z_i^* = P_i^*$ for each~$i$ (and substituting the empty word with the $n \times n$ identity operator).
    Given also a matrix polynomial $p = \sum_w a_w w$, we define the \emph{evaluation of $p$ at~$\calL$} to be the following operator on $\C^{n} \otimes \C^r$:\footnote{Note that coefficients $a_w$ are written on the left in~$p$, as is conventional, but we take the tensor/Kronecker product on the right so that the matrix form of $p(\calL)$ may be more naturally regarded as an $n \times n$ matrix with $r \times r$ entries.}
    \[
        p(\calL) = \sum_w \calL^w \otimes a_w.
    \]
\end{definition}
\begin{remark} \label{rem:self-adj}
        Note that each $P_i$ is unitary  and each $M_i$ a unitary involution (as promised), so $p^*(\calL) = p(\calL)^*$.
    Thus $p(\calL)$ is a self-adjoint operator whenever $p^*$ is a self-adjoint polynomial.
    In this case we also have that $p(\calL)$ may be viewed as the adjacency matrix of an undirected graph on vertex set~$[n]$ with $r \times r$ edge-weights.
\end{remark}

Note that the evaluation $p(\calL)$ of a matrix polynomial may be viewed as the adjacency matrix of a undirected graph on~$[n]$ with $r \times r$ edge-weights; or, its extension may be viewed as the adjacency matrix of an undirected graph on $[n] \times [r]$ with scalar edge-weights.
In this way, each fixed matrix polynomial~$p$, when applied to a random lift, gives rise to a random (undirected, scalar-weighted) graph model.
\begin{examples}
	A simple example is the matrix polynomial
    \[
        p(Y_1,Y_2,Y_3) = Y_1 + Y_2 + Y_3.
    \]
    Here $r = 1$ and each coefficient is just the scalar~$1$.
    This~$p$ gives rise to a model of random edge-signed $3$-regular graphs on~$[n]$.
\end{examples}

By moving to actual matrix coefficients with $r > 1$, one can get the random (signed) graph model given by randomly $n$-lifting any base $r$-vertex graph~$H$.
\begin{examples}
As a simple example,
\[
    p(Z_1, Z_2, Z_3) = \begin{pmatrix} 0 & 1 \\ 0 & 0 \end{pmatrix} Z_1 + \begin{pmatrix} 0 & 0 \\ 1 & 0 \end{pmatrix} Z_1^* + \begin{pmatrix} 0 & 1 \\ 0 & 0 \end{pmatrix} Z_2 + \begin{pmatrix} 0 & 0 \\ 1 & 0 \end{pmatrix} Z_2^* +\begin{pmatrix} 0 & 1 \\ 0 & 0 \end{pmatrix} Z_3 + \begin{pmatrix} 0 & 0 \\ 1 & 0 \end{pmatrix} Z_3^*
\]
is the recipe for random $3$-regular (edge-signed) $(n + n)$-vertex bipartite graphs.
The reader may like to view this as a $2 \times 2$ matrix of polynomials,
\[
    p(Z_1, Z_2, Z_3) = \begin{pmatrix} 0 & Z_1 + Z_2 + Z_3 \\ Z_1^* + Z_2^* + Z_3^* & 0 \end{pmatrix},
\]
but recall that we actually Kronecker-product the coefficient matrices ``on the other side''.
So rather than as a $2 \times 2$ block-matrix with $n \times n$ blocks, we think of the resulting adjacency matrix as an $n \times n$ block-matrix with $2 \times 2$ blocks; equivalently, an $n$-vertex graph with $2 \times 2$ matrix edge-weights.
\end{examples}
\begin{examples}
    The matrix polynomial $p_{333}$ mentioned in \eqref{eqn:p333} gives an example of a \emph{nonlinear} polynomial with matrix coefficients.
    Again, we wrote it there as a $3 \times 3$ matrix of polynomials for compactness, but for analysis purposes we will view it as a degree-$2$ polynomial with $3 \times 3$ coefficients.
\end{examples}

\begin{definition}[$\infty$-lift]
    Formally, we extend \Cref{def:n-lift} to the case of $n = \infty$ as follows.
    Let $V_\infty$ denote the free product of groups $\Z_2^{\star d} \star \Z^{\star e}$, with its components generated by $g_1, \dots, g_d$, $h_1, h_1^{-1}, \dots, h_e, h_e^{-1}$.
    Thus the elements of~$V_\infty$ are in one-to-one correspondence with the reduced words over indeterminates $Y_1, \dots, Z_{e}^*$.
    The generators $g_1, \dots, g_d, h_1, \dots, h_e$ act as permutations on~$V_\infty$ by left-multiplication, with the first~$d$ in fact being matchings.
    We write $\sigma_1, \dots, \sigma_{d+e}$ for these permutations, and we also identify them with their associated permutation operators on $\ell_2(V_\infty)$.
    Finally, we write $\lift_\infty = (\sigma_0, \dots, \sigma_{d+2e})$ for \emph{``the'' $\infty$-lift} associated to~$p$.
    (Note that this lift is ``unsigned''.)
\end{definition}

\begin{definition}[Evaluation at the $\infty$-lift, and $\mathfrak{X}_p$.]
    The evaluation of a matrix polynomial~$p$ at the infinity lift $\calL_\infty$ is now defined just as in \Cref{def:eval}; the resulting operator $p(\calL_\infty)$ operates on $\ell_2(V_\infty) \otimes \C^r$.
    We may think of the result as a matrix-weighted graph on vertex set~$V_\infty$, and we will sometimes denote this graph by~$\mathfrak{X}_p$.
    When $p$ is understood, we often write $A_\infty = p(\calL_\infty)$ for the adjacency operator of~$\mathfrak{X}_p$, which can be thought of as an infinite matrix with rows/columns indexed by~$V_\infty$ and entries from~$\C^{r \times r}$, or as its ``extension'' $\wt{A}_\infty$, an infinite matrix with rows/columns indexed by $V_\infty \times [r]$ and scalar entries.
\end{definition}
\begin{examples}
	For the polynomial $p = Y_1 + \cdots + Y_d$, the corresponding graph $\mathfrak{X}_p$ is the infinite (unweighted) $d$-regular tree.
\end{examples}

We may now state a theorem which is essentially the main result (``Theorem~2'') of~\cite{BC19}.
The small difference is that our notion of random $n$-lifts, which includes $\pm 1$ \emph{signs} on the matchings/permutations, lets us eliminate mention of ``trivial'' eigenvalues (see~\cite[Thms.~1.9,~10.10]{OW20}).
\begin{theorem} \label{thm:bc-with-edge-signing}
        Let $p$ be a self-adjoint matrix polynomial with coefficients from $\C^{r \times r}$ on indeterminates $Y_1, \ldots, Z_e^*$.
        Then for all $\eps, \beta > 0$ and sufficiently large~$n$, the following holds:

        Let $\bA_n = p(\bcalL_n)$, where $\bcalL_n$ is a random $n$-lift, and let $A_\infty = p(\calL_\infty)$.
        Then except with probability at most~$\beta$, the spectra $\spec(\bA_n)$ and $\spec(A_\infty)$ are at Hausdorff distance at most~$\eps$
\end{theorem}

\subsection{Random lifts as optimization problems}
Given a Hermitian (i.e., self-adjoint) matrix $A \in \C^{n \times n}$, we are interested in the task of maximizing $x^\top A x$ over all Boolean vectors $x \in \{\pm 1\}^n$.
(Since $A$ is Hermitian, the quantity $x^\top A x$ is always real, so this maximization problem makes sense.)
This is the same as maximizing the homogeneous degree-$2$ (commutative) polynomial $\sum_{i,j} A_{ij} x_i x_j$ over $x \in \{\pm 1\}^n$, and it is also essentially the same task as the Max-Cut problem on (scalar-)weighted undirected graphs.
More precisely, if $G$ is a weighted graph on vertex set~$[n]$ with adjacency matrix~$A$, then $G$'s maximum cut is indicated by the $x \in \{\pm 1\}^n$ that maximizes $x^\top (-A) x$.
For the sake of scaling we will also include a factor of~$\frac1n$ in this optimization problem, leading to the following definition:
\begin{definition}[Optimal value]
	Given a Hermitian matrix $A \in \C^{n \times n}$, we define
	\[
	 \OPT(A) = \sup_{x \in \{\pm 1 \}^n} \braces*{\tfrac{1}{n} x^\transp A x} = \sup_{x \in \left\{\pm \frac{1}{\sqrt{n}} \right \}^n} \braces*{x^\transp A x}.
	\]
    (For finite-dimensional~$A$, the $\sup$s and $\inf$s mentioned in this section are all achieved.)
\end{definition}
We remark that
\[
     x^\transp A x = \tr(x^\transp A x) = \tr(x x^\top A) = \la x x^\top, A \ra,
\]
where we use the notation $\la B,C \ra = \tr(B C)$.  Thus we also have
\[
    \OPT(A) = \sup_{\rho \in \mathrm{Cut}_n} \braces*{\tfrac1n\la \rho, A \ra},
\]
where $\mathrm{Cut}_n$ is the ``cut polytope'', the convex hull of all matrices of the form $x x^\transp$ for $x \in \{\pm 1\}^n$.
(Since $\la \rho, A \ra$ is linear in $\rho$, maximizing over the convex hull is the same as maximizing over the extreme points, which are just those matrices of the form~$x x^\transp$.)

The above optimization problem has a natural relaxation: maximizing $\tfrac1n x^\transp A x$ over all \emph{unit} vectors~$x$.
This leads to the following efficiently computable upper bound on $\OPT(A)$:
\begin{definition}[Eigenvalue bound]    \label{def:eig-bound}
	Given a Hermitian matrix $A \in \C^{n \times n}$, we define the \emph{eigenvalue bound} to be
	\[
    	\Eig(A) = \sup\{\la \rho, A \ra : \rho \pgeq 0,\ \tr(\rho) = 1\},
	\]
    where here $\rho \pgeq 0$ denotes that $\rho$ is (Hermitian and) positive semidefinite.
\end{definition}
The matrices $\rho$ being optimized over in~$\Eig(A)$ are known as \emph{density matrices}; i.e., $\Eig(A)$ is the maximal inner product between~$A$ and any density matrix.
Note that if $\varrho \in \mathrm{Cut}_n$, then $\rho = \tfrac1n \varrho$ is a density matrix.
Thus, $\Eig(A)$ is a relaxation of $\OPT(A)$, or in other words, $\OPT(A) \leq \Eig(A)$.

The set of density matrices is convex, and it's well known that its extreme points are all the rank-$1$ density matrices; i.e., those $\rho$ of the form $xx^\transp$ for $x \in \C^{n}$ with $\|x\|_2^2 = 1$.
Thus in $\Eig(A)$ it is equivalent to just maximize over these extreme points:
	\[
		\Eig(A) = \sup_{\substack{x \in \C^n \\ \|x\|^2_2 = 1}} \braces*{\langle xx^\transp, A \rangle}
                   = \sup_{\substack{x \in \C^n \\ \|x\|^2_2 = 1}} \braces*{x^\transp A x}.
	\]
From this formula we see that $\Eig(A)$ is also equal to $\lambda_{\max}(A)$, the maximum eigenvalue of~$A$; hence the terminology ``eigenvalue bound''.
One may also think of  $\Eig(A)$ and $\lambda_{\max}(A)$ as SDP duals of one another.

We now mention another well known, tighter, upper bound on $\OPT(A)$.
\begin{definition}[Basic SDP bound]
	Given a Hermitian matrix $A \in \C^{n \times n}$, the \emph{basic SDP bound} is defined to be
	\[
		\GW(A) = \sup \{\la \rho, A \ra : \rho \pgeq 0,\ \rho_{ii} = \tfrac{1}{n}, \forall i\}.
	\]
\end{definition}
Recall that an $n \times n$ matrix $\varrho$ is a \emph{correlation matrix}~\cite{Sty73} if it is PSD and has all diagonal entries equal to~$1$.
Thus $\GW(A)$ is equivalently maximizing $\tfrac1n \la \varrho, A \ra$ over all correlation matrices~$\varrho$.
We also note that any cut matrix is a correlation matrix, and any correlation matrix is a density matrix, hence so
\[
    \OPT(A) \leq \GW(A) \leq \Eig(A).
\]

\begin{definition}[Dual SDP bound]
    The semidefinite dual of $\SDP(A)$ is the following~\cite{DP93}:
    \[
    	\GWDual(A) = \inf_{\substack{\oldu \in \mathbb{R}^n \\ \avg(\oldu_1, \ldots, \oldu_n) = 0}} \{\lambda_{\max}(A + \diag(\oldu))\}.
    \]
\end{definition}
Despite the fact that the usual ``Slater condition'' for strong SDP duality fails in this case (because the set of correlation matrices isn't full-dimensional), one can still show~\cite{PR95} that $\GW(A) = \GWDual(A)$ indeed holds for finite-dimensional~$A$.

\begin{remark}
    In this work we frequently consider matrix-weighted graphs with adjacency matrices~$A$, thought of as $n \times n$ matrices with entries from~$\C^{r \times r}$.
    For such matrices, whenever we write $\Opt(A)$, we mean $\Opt(\wt{A})$ for the $nr \times nr$ ``extension'' matrix~$\wt{A}$ (see \Cref{def:extension}), and similarly for  $\Eig(A)$, $\lambda_{\max}(A)$, $\GW(A)$, $\GWDual(A)$.
\end{remark}

As mentioned in \Cref{sec:intro}, the eigenvalue bound $\lambda_{\max}(A)$ makes sense when $A$ is the adjacency matrix of an infinite graph (with bounded degree).
However $\GW(A)$ does not extend to the infinite case, as the number ``$n$'' appearing in its definition is not finite.
On the other hand, we now introduce a new, intermediate, ``maximum eigenvalue-like'' bound that is appropriate for matrix-weighted graphs.
This is the ``\geneig bound'' appearing in the statement of our main \Cref{thm:ourmain}.
In the following \Cref{sec:sdp-duality}, we will show that it generalizes well to the case of infinite graphs.
\begin{definition}[\Geneig bound]
    Let $A$ be an $n \times n$ Hermitian matrix with entries from~$\C^{r \times r}$.
    We define its \emph{\geneig bound} to be
	\[
		\GenEig(A) = \sup \{\la \rho, A \ra  : \rho \pgeq 0,\ \tr(\rho) \in \tfrac1r \mathrm{Corr}_r\},
	\]
    where:
    \begin{itemize}
        \item the matrices $\rho$ are also thought of as $n \times n$ matrices with entries from~$\C^{r \times r}$;
        \item $\la \rho, A \ra$ is interpreted as $\la \wt{\rho}, \wt{A} \ra$;
        \item $\tr(\rho)$ denotes the sum of the diagonal entries of~$\rho$, which is an $r \times r$ matrix;
        \item $\mathrm{Corr}_r$ is the set of $r \times r$ correlation matrices;
        \item in other words, the final condition is that $\tr(\rho)_{ii} = \frac1r$ for all $i \in [r]$.
    \end{itemize}
\end{definition}

\begin{remark}
    As mentioned, the \geneig bound can be viewed as ``intermediate'' between the eigenvalue bound and the SDP bound.
    To explain this, suppose $A$ is an $n \times n$ Hermitian matrix.
    On one hand, we can regard $A$ as an $n \times n$ matrix with $1 \times 1$ matrix entries ($r = 1$); in this viewpoint, $\GenEig(A) = \Eig(A)$.
    On the other hand, we can regard $A$ as a $1 \times 1$ matrix with a single $n \times n$ matrix entry ($r = n$); in this viewpoint, $\GenEig(A) = \GW(A)$.
\end{remark}

It is easy to see that the \geneig bound indeed has an SDP formulation, and we now state its SDP dual:
\begin{definition}
    The SDP dual of $\GenEig(A)$ is the following:
	\[
		 \GenEigDual(A) = \inf_{\substack{\oldu \in \mathbb{R}^r \\ \avg(\oldu_1, \ldots, \oldu_r) = 0}} \{\lambda_{\max}(A + \Id_{n \times n} \otimes \diag(\oldu) )\}.
	\]
\end{definition}
Weak SDP duality, $\GenEig(A) \leq \GenEigDual(A)$, holds as always, but again it is not obvious that strong SDP duality holds.
In fact, not only does strong duality hold, it even holds in the case of \emph{infinite} matrices~$A$.
This fact is crucial for our work, and proving it the subject of the upcoming technical section.

\section{The infinite SDPs}   \label{sec:sdp-duality}

This technical section has two goals.
First, in \Cref{sec:dual2} we show that strong duality holds with $\GenEig(A) = \GenEigDual(A)$, even for infinite matrices~$A$ with $r \times r$ entries.
Even in the finite case this is not trivial, as the feasible region for the SDP $\GenEig(A)$ is not full-dimensional, and hence the Slater condition ensuring strong duality does not apply.
The infinite case involves some additional technical considerations, needed so that we may eventually apply the strong duality theorem for conic linear programming of Bonnans and Shapiro~\cite[Thm.~2.187]{BS00}.
Second, in \Cref{sec:dual-finite}, we show that in the optimization problem $\GenEig(A)$, values arbitrarily close to the optimum can be achieved by matrices~$\rho$ of finite support (i.e., with only finitely many nonzero entries).
Indeed (though we don't need this fact), these finite-rank~$\rho$ need only have rank at most~$r$.
This fact is familiar from the case of $r = 1$, where the optimizer in the eigenvalue bound \Cref{def:eig-bound} is achieved by a $\rho$ of rank~$1$ (namely $\ketbra{\psi}{\psi}$ for any maximum eigenvector~$\ket{\psi}$).
Finally, in \Cref{sec:dual-wrapup} we consolidate all these results into a theorem statement suitable for use with graphs produced by infinite lifts of matrix polynomials.

\subsection{SDP duality} \label{sec:dual2}
Let $V$ be a countable set and write $\calH = \ell_2(V)$ for the associated (complex, separable) Hilbert space of square-summable functions $f : V \to \C$.
We write $B_{00}(\calH)$,  $B_1(\calH)$, $B(\calH)$ for the spaces of finite-rank, trace-class, and bounded operators on~$\calH$, respectively.
Focusing on $B_1(\calH)$ (with the weak topology) and $B(\calH)$ (with the $\sigma$-weak topology), these are locally convex topological vector spaces forming a dual pair with bilinear map $\la \cdot, \cdot \ra : B_1(\calH) \times B(\calH) \to \C$ defined by $\la \rho, a \ra = \tr(\rho a)$ (see~\cite[Thm.~VI.26]{RS80}, \cite[Thm.~B.146]{Lan17}, or~\cite[Prop.~2.4.3]{BR02}).
Recall~\cite[Lem.~B.142]{Lan17} that $\la \rho, a \ra \leq \|\rho\|_1 \|a\|$, where $\|\rho\|_1 = \tr \sqrt{\rho^* \rho}$ is the trace-norm.

Write $B_1(\calH)_{\mathrm{sa}}$ (respectively, $B(\calH)_{\mathrm{sa}}$) for the (closed) real subspace of self-adjoint operators in~$B_1(\calH)$ (respectively, $B(\calH)_{\mathrm{sa}}$); note that $B_1(\calH)_{\mathrm{sa}}$ and $B(\calH)_{\mathrm{sa}}$ continue to form a dual pair (see, e.g.,~\cite[p.~212]{Mey94}).
Recall that $a \in B(\calH)_{\mathrm{sa}}$ is positive semidefinite if and only if $\braket{\varphi|a|\varphi}\geq 0$ for all $|\varphi\rangle \in \calH$ (for such~$a$ we have $\sqrt{a^* a} = a$); as usual we write $b \pgeq a$ to mean that $b-a$ is positive semidefinite.
Letting $B_1(\calH)_{+}$ (respectively, $B(\calH)_{+}$)  denote the positive semidefinite operators in $B_1(\calH)_{\mathrm{sa}}$ (respectively, $B(\calH)_{\mathrm{sa}}$), we have that these are both (nonempty) closed, convex cones (for  $B(\calH)_{+}$, see \cite[Prop.~3.7]{Con90}, \cite[Prop.~C.51]{Lan17}; for $B_1(\calH)_{+}$ see \cite[Sec.~4]{Fri19},  \cite[App.~A.2]{DE20}, \cite[Sec.~2]{SH08}); further, they are topologically dual cones~\cite[App.~A.2]{DE20}.

\paragraph{Our SDP.}
We now introduce an SDP that is equivalent to our \geneig SDP; however, we will express it with scalar entries (rather than matrix entries).
Fix any $a \in B(\calH)_{\mathrm{sa}}$.
Let $V = V_1 \sqcup V_2 \sqcup \cdots \sqcup V_r$ be a partition of~$V$ into $r \in \N_+$ nonempty parts, and for $1 \leq j \leq r$ let $I_j \in B(\calH)_{\mathrm{sa}}$ be the operator $\sum_{v \in V_j} \ketbra{v}{v}$.
Consider the following semidefinite program:
\begin{equation}
    \sup_{\rho \in B_1(\calH)_{+}} \la \rho, a \ra  \quad \text{subject to}\quad \la \rho, I_j \ra = \tfrac1r,\ j = 1 \dots r.\tag{SDP-P}
\end{equation}
To relate this back to our definition of $\GenEig(A)$ (for infinite matrices), suppose that~$A$ is a self-adjoint matrix, indexed by countable vertex~$V_\infty$, with $r \times r$ entries (and only finitely many nonzero entries per row/column).
Let $a = \wt{A}$ be its ``extension'', with rows/columns indexed by $V = V_\infty \times [r]$, and let $V_j = V_\infty \times \{j\}$.
Then (SDP-P) is precisely how we define $\GenEig(A)$.

We remark that (SDP-P) is always feasible.
By summing the constraints on~$\rho$ we get $\tr(\rho) = 1$; i.e., $\rho$ is required to be a density operator.
In particular, $\la \rho, a \ra \leq \tr(\rho) \|a\|$ always so the optimum value of~(SDP-P) is finite.

We would like to show that the optimum value of (SDP-P) is equal to that of the following dual semidefinite program which, in the setup mentioned above, is equivalent to our definition of $\GenEigDual(A)$:
\begin{equation}
    \inf_{\oldu \in \R^r} \lambda_{\max}(a + \oldu_1 I_1 + \cdots + \oldu_r I_r) \quad \text{subject to}\quad \avg(\oldu_1, \dots, \oldu_r)  = 0. \tag{SDP-D}
\end{equation}
Showing this will take a couple of steps.
The first is to mechanically write down the Lagrangian dual of (SDP-P), which is:
\begin{align}
    \inf_{y \in \R^r} \tfrac1r \sum_{j}y_j  \quad \text{subject to}\quad a \pleq y_1 I_1 + \cdots + y_r I_r.\tag{D1}
\end{align}
We first show that (D1) is equivalent to (SDP-D).
We may reparameterize all $y \in \R^r$ by defining $\ol{y} = \tfrac1r \sum_j y_j$ and writing $y_j = \ol{y} - \tfrac1r$; in this way, every $y \in \R^r$ corresponds to a pair $\ol{y} \in \R$ and $\oldu \in \R^r$ satisfying $\avg(\oldu_1, \dots, \oldu_r) = 0$, and vice versa.
Under this reparameterization, (D1)~becomes
\begin{align}
    \inf_{\substack{\ol{y} \in \R \\ \oldu \in \R^r }} \ol{y}   \quad \text{subject to} \quad \avg(\oldu_1, \dots, \oldu_r) = 0 \text{ and } a \pleq \ol{y} I - \sum_j \oldu_j I_j .\tag{D2}
\end{align}
The second constraint here is equivalent to $\lambda_{\max}(a + \sum_j \oldu_j I_j) \leq \ol{y}$,  and hence the optimal choice of $\ol{y}$ given $\oldu$ is achieved by $\lambda_{\max}(a + \sum_j \oldu_j I_j)$.
Thus (D1) is equivalent to (D2) is equivalent to (SDP-D), as claimed.

We would next like to claim that strong duality holds for the dual pair (SDP-P) and~(D1); however, there is a difficulty because the feasible region of (SDP-P) only has nonempty \emph{relative interior}, not interior. Alternatively, one might say the difficulty is that the feasible region for~(D1) is not bounded.
However we can fix this by introducing the following equivalent bounded variant:
\begin{equation}
    \inf_{y \in \R^r} \tfrac1r\sum_jy_j  \quad \text{subject to}\quad a \pleq y_1 I_1 + \cdots + y_r I_r \text{ and } |y_j| \leq C,\ j = 1 \dots r,\tag{D3}
\end{equation}
where $C = 2r\|a\|$.

\begin{claim}
    \textnormal{(D3)} is equivalent to \textnormal{(D1)}.
\end{claim}
\begin{proof}
    Suppose we have any feasible solution $y$ for~(D1).
    First observe that for any $1 \leq j \leq r$, if we take an arbitrary
    $v_j \in V_j$ we have
    \begin{equation} \label[ineq]{ineq:lower}
        y_1 I_1 + \cdots + y_r I_r  \pgeq a \implies y_j \geq \braket{v_j | a | v_j} \geq -\|a\|.
    \end{equation}
    Next observe that the optimum value of~(D1) is at most~$\|a\|$ (since we could take $y_j = \|a\|$ for all~$j$); thus we may restrict attention to $y$'s that achieve objective value at most~$\|a\|$.
    But now we claim any such feasible~$y$ will have $y_k \leq 2r\|a\|$ for any~$1 \leq k \leq r$.
    To see this, observe (using \Cref{ineq:lower}) that
    \[
        \tfrac1r \sum_j y_i \geq \tfrac1r (y_k - \sum_{j \neq k}\|a\|) \geq y_k/r - \|a\|.
    \]
    Thus if $y_k > 2r\|a\|$, it is strictly not optimal.
    We conclude that we may add the constraints $-\|a\| \leq y_j \leq 2r\|a\|$ for~$j = 1 \dots r$ to~(D1) without changing it.
    Clearly it doesn't hurt to widen this interval to $-C \leq y_j \leq C$ (for the sake of symmetry), and so we conclude that (D3)~is indeed equivalent to~(D1).
\end{proof}

We can now write the Lagrangian dual of~(D3), which is
\begin{equation}
    \sup_{\rho \in B_1(\calH)_{+}} \la \rho, a \ra - C\sum_{j=1}^r \abs{\la \rho, I_j \ra - \tfrac1r}.  
    \tag{P3}
\end{equation}
\begin{claim}   \label{claim:P-equivalent}
    \textnormal{(P3)} is equivalent to \textnormal{(SDP-P)}.
\end{claim}
\begin{proof}
    Clearly the objective value of (SDP-P) is at most that of~(P3), so it suffices to show the converse inequality.
    To that end, let $\rho$ be any feasible solution for~(P3), and write the objective value as $o(\rho) - p(\rho)$, where $o(\rho) = \la \rho, a \ra$ (``objective'') and $p(\rho) = C \sum_j \abs{\langle \rho, I_j \rangle - \tfrac1r}$ (``penalty'').

    Let $\eps = (\max_j\{r\la \rho, I_j \ra\} - 1)_+$, so that $\la \rho, I_j \ra \leq \tfrac1r(1+\eps)$ for all~$j$, and write $\wt{\rho} = \rho/(1+\eps) \in B_1(\calH)_+$.
    Note that $p(\wt{\rho}) = C \sum_j \delta_j$, where $\delta_j = \tfrac1r - \la \wt{\rho}, I_j \rangle \geq 0$.
    Finally, define $\rho' = \wt{\rho} + \sum_j \delta_j\!\ketbra{v_j}{v_j} \in B_1(\calH)_+$, where each $v_j$ is an arbitrary element of~$V_j$.
    By construction, $p(\rho') = 0$.
    We now want to show
    \begin{equation}    \label[ineq]{ineq:showme}
        o(\rho') \geq o(\rho) - p(\rho) \quad\iff\quad p(\rho) \geq o(\rho) - o(\rho').
    \end{equation}
    We will then have that $\rho'$ is feasible for (SDP-P), with an objective value at least that of~$\rho$'s in~(P3); this will complete the proof.

    First, we have
    \[
        p(\rho) = C \sum_j \abs{\langle \rho, I_j \rangle - \tfrac1r} = C \sum_j \abs{ (1+\eps)\la\wt{\rho}, I_j \ra - \tfrac1r} = 2r\|a\|\sum_j \abs{\eps/r - (1+\eps)\delta_j}.
    \]
    On the other hand, we have
    \[
        o(\rho) - o(\rho') = \la \rho, a \ra - \tfrac{1}{1+\eps}\la \rho, a \ra - \delta_j \sum_j \braket{v_j|a|v_j} \leq \tfrac{\eps}{1+\eps} \la \rho, a \ra + \|a\|\sum_j \delta_j \leq \|a\|\parens{\tfrac{\eps}{1+\eps} \tr \rho + \sum_j \delta_j}.
    \]
    Summing the inequality $\la \rho, I_j \ra \leq \tfrac1r(1+\eps)$ yields $\tr \rho \leq 1+\eps$, so to establish \Cref{ineq:showme} it remains to show
    \begin{equation}    \label[ineq]{ineq:showme2}
        2r\sum_j \abs{r\eps - (1+\eps)\delta_j} \geq \eps +  \sum_j \delta_j \quad\iff\quad \sum_j\left(2r \abs{r\eps - (1+\eps)\delta_j} - \delta_j\right) \geq \eps.
    \end{equation}
    This is certainly true if $\eps = 0$, so it remains to assume $\eps > 0$.
    Then writing $j_0$ for the $j$ achieving $\max_j \{r\la \rho, I_j \ra\}$, we have $\delta_{j_0} = 0$ by construction.
    This gives us a contribution of $2\eps$ to the sum on the right in \Cref{ineq:showme2}.
    On the other hand, for $j \neq j_0$, the worst possible value for~$\delta_j$ is $\frac{\eps}{r(1+\eps)}$, in which case the contribution to the sum is $-\frac{\eps}{r(1+\eps)}$; summing this over all $j \neq j_0$ yields at worst $-\frac{\eps}{1+\eps}$.
    Thus we have shown the sum in \Cref{ineq:showme2} is at least $2\eps - \frac{\eps}{1+\eps}  \geq \eps$, as needed.
\end{proof}

By inspecting the above proof we may extract the following lemma, which will be useful later:
\begin{lemma}                                       \label{lem:fix-sdp}
    Suppose $\rho \in B_1(\calH)_+$ has $|\la \rho, I_j \ra - \tfrac1r| \leq \eta$ for each $1 \leq j \leq r$.
    Then there is $\rho' \in B_1(\calH)_+$ satisfying $\la \rho', I_j \ra = \tfrac1r$ for each~$j$ and with $\abs{\la \rho, a \ra - \la \rho', a \ra} \leq Cr \eta$.
    Furthermore, if $\rho$ has finite rank (respectively, support) then so too does~$\rho'$.
\end{lemma}

Returning to (P3), formally speaking we are regarding it as the conic linear program
\begin{equation}
    \sup_{(\rho, \lambda^+, \lambda^-) \in B_1(\calH)_{+} \oplus \R^r_{\geq 0} \oplus \R^r_{\geq 0}} \la \rho, a \ra - C\sum_{j=1}^r (\lambda^+_j + \lambda^-_j) \quad \text{subject to}\quad \la \rho, I_j \ra - \tfrac1r = \lambda^+_j - \lambda^-_j,\ j = 1 \dots r.
    \tag{P3}
\end{equation}
Not only is this feasible, meaning the vector  $(\tfrac1r, \dots, \tfrac1r)$ is in the region
\[
    D = \Bigl\{\bigl(\la \sigma, I_1 \ra - (\lambda_1^+ - \lambda_1^-), \dots, \la \sigma, I_r \ra  - (\lambda_1^+ - \lambda_1^-)\bigr) : \sigma \in B_1(\calH)_{+},\ \lambda^+ \in \R^r_{\geq 0},\ \lambda^- \in \R^r_{\geq 0} \Bigr\} \subseteq \R^r,
\]
but it is even in the interior of the region, since $D$ is in fact all of~$\R^r$.
We are therefore finally in a position to apply a strong duality theorem for conic linear programming~\cite[Thm.~2.187]{BS00} to the pair (P3),~(D3).
We conclude:
\begin{theorem}
    \textnormal{(SDP-P)} and \textnormal{(SDP-D)} have a common value~$s^*$, and $s^*$ is achieved in~\mbox{\textnormal{(SDP-D)}} by some $\hat{\oldu} \in \R^r$.
    We may therefore write $s^* = \lambda_{\max}(\hat{a})$, where $\hat{a} = a + \sum_j \hat{\oldu}_j I_j$.
\end{theorem}

\subsection{Nearly optimal finite, rank-$r$ solutions} \label{sec:dual-finite}
Although the optimal value~$s^*$ is achieved in~(SDP-D) (that is, its ``$\inf$'' may be written ``$\min$''), the value~$s^*$ might not be achieved in~(SDP-P).
However for any $\delta > 0$, we can find $\rho \in B_1(\calH)_{+}$ satisfying $\la \rho, I_j \ra = \tfrac1r$ for all~$j$ and with $\la \rho, a \ra \geq s^* - \delta$.
Let us further simplify this~$\rho$.

\paragraph{Finite rank.}
First, recall that $B_{00}(\calH)_+$ (the finite-rank positive semidefinite operators on~$\calH$) is dense in~$B_1(\calH)_+$ with respect to the trace-norm (see \cite[Thm.~1.11(d)]{Con85}, with the proof clearly holding under the positive semidefinite restriction).
Thus we can find a finite-rank $\rho_0 \in B_{00}(\calH)_{+}$ satisfying $\abs{\la \rho_0, I_j \ra - \tfrac1r} \leq \delta$ for all~$j$ and with $\la \rho_0, a \ra \geq s^* - \delta -  \|a\| \delta$.
Next, using \Cref{lem:fix-sdp} we can convert this to another finite-rank $\rho_1 \in B_{00}(\calH)_{+}$ satisfying $\la \rho_1, I_j \ra = \tfrac1r$ for all~$j$ and with $\la \rho_1, a \ra \geq s^* - \delta_1$, where $\delta_1 = \delta + \|a\|\delta + C \delta \to 0$ as $\delta \to 0$.

\paragraph{Finite support.}
Since $\rho_1 \in B_{00}(\calH)_{+}$ has trace~$1$ we can write $\rho_1 = \sum_{i=1}^{k_1} \lambda_i \ketbra{\varphi_i}{\varphi_i}$ for some $\lambda = (\lambda_1, \dots, \lambda_{k_1}) \in \R_+^{k_1}$ forming a probability distribution on~$[k_1]$ and some orthonormal vectors $\ket{\varphi_i} \in \ell_2(V)$.
We can approximate each $\ket{\varphi_i}$ by a unit vector $\ket{\varphi'_i}$ of finite support satisfying $\abs{\braket{\varphi_i|\varphi'_i}}^2 \geq 1-\delta_1$, from which it follows (an easy calculation~\cite[(9.60), (9.99)]{NC10}) that $\|e_i\|_1 \leq 2\sqrt{\delta_1}$, where $e_i = \ketbra{\varphi_i}{\varphi_i} - \ketbra{\varphi'_i}{\varphi'_i}$.
Now let $\rho_2 = \sum_{i=1}^{k_1} \lambda_i \ketbra{\varphi'_i}{\varphi'_i}$, a positive semidefinite matrix of finite support.
Then
\[
    \la \rho_2, a \ra = \la \rho_1, a \ra - \sum_{i=1}^{k_1} \lambda_i \la e_i, a \ra \geq s^* - \delta_1 - \sum_{i=1}^{k_1}  \lambda_i \cdot 2\sqrt{\delta_1} \|a\| = s^* - \delta_1 - 2\sqrt{\delta_1} \|a\|,
\]
and one can similarly check that $|\la \rho_2, I_j \ra - \tfrac1r| \leq 2\sqrt{\delta_1}$ for all~$j$.
Again applying \Cref{lem:fix-sdp}, we can convert this to another finite-support $\hat{\rho} \in B_{00}(\calH)_{+}$ satisfying $\la \hat{\rho}, I_j \ra = \tfrac1r$ for all~$j$ and with $\la \rho_1, a \ra \geq s^* - \delta_2$, where $\delta_2 = \delta_1 + 2\sqrt{\delta_1} \|a\| +2C\sqrt{\delta_1} \to 0$ as $\delta \to 0$.
Since $\delta_2$ can be made arbitrarily small by taking $\delta \to 0$, we finally conclude:
\begin{proposition}
    For any $\eps > 0$, there is a feasible finite-support solution $\hat{\rho}$ to~\textnormal{(SDP-P)} achieving objective value at least $s^* - \eps = \lambda_{\max}(\hat{a}) - \eps$.
\end{proposition}

\paragraph{Rank~$\mathbf{r}$.}
In general, any solution~$\rho$ to (SDP-P) need not have rank more than~$r$.
Although we don't strictly need this fact, it is particularly easy to show for~$\rho$ of finite support, and we do so now.
\begin{proposition}   \label{prop:rank-r}
    Let $\hat{\rho}$ be a feasible solution to \textnormal{(SDP-P)}, with nonzero entries only in rows or columns indexed by a finite set $F \subseteq V$.
    Then there is another $\tilde{\rho}$ feasible for \textnormal{(SDP-P)}, supported on~$F$ and with rank at most~$r$ such that $\la \tilde{\rho}, a \ra \geq \la \hat{\rho}, a \ra$.
\end{proposition}
\begin{proof}
    We think of $\hat{\rho}$ as a matrix indexed just by~$\hat{F}$, and write $a_{\hat{F}}$ for the submatrix of $a$ on rows/columns indexed by~$\hat{F}$.
    Suppose $\hat{\rho}$ has eigenvalues $\mu_1, \ldots, \mu_{|\hat{F}|} \geq 0$ (nonnegative, since $\hat{\rho} \geq 0$), with corresponding orthonormal eigenvectors $\ket{\psi_1}, \ldots, \ket{\psi_{|\hat{F}|}}$.
    Then $\hat{\rho}$ achieves objective value
    \[
        c \coloneqq \la \hat{\rho}, a \ra = \la \hat{\rho}, a_{\hat{F}} \ra = \tr(\hat{\rho} a_{\hat{F}}) = \tr\parens*{\parens*{\sum_i \mu_i \ketbra{\psi_i}{\psi_i}} a_{\hat{F}}} = \sum_i \mu_i \braket{\psi_i|a_{\hat{F}}|\psi_i}.
    \]
    Also, writing $I_j$ for the identity matrix restricted to $V_j \cap \hat{F}$, feasibility of~$\rho$ implies
    \[
        \tfrac1r = \la \hat{\rho}, I_j \ra = \sum_i \mu_i \braket{\psi_i|I_j|\psi_i}, \quad j = 1 \dots r.
    \]
    Now if we imagine $\ket{\psi_1}, \ldots, \ket{\psi_{|\hat{F}|}}$ are fixed, and $\lambda_1, \dots, \lambda_{|\hat{F}|}$ are real variables, the above tells us that the following linear program ---
    \begin{equation}
        \max_{\lambda \in \R^{\geq 0}_{|\hat{F}|}} \left\{ \sum_i \lambda_i \braket{\psi_i|a_{\hat{F}}|\psi_i} : \sum_i \lambda_i \braket{\psi_i|I_j|\psi_i} = \tfrac1r,\ j = 1 \dots r\right\}
        \tag{LP}
    \end{equation}
    --- has optimal value at least~$c$, since we may take $\lambda = \mu$.
    By standard theory of finite linear programs, the optimal value occurs at a vertex where $|\hat{F}|$ linearly independent constraints are tight.
    Since there are only $r$ equality constraints in~(LP), there is an optimal solution~$\lambda^*$ with at least $|\hat{F}|-r$ zero entries.
    Thus $\tilde{\rho} = \sum_i \lambda_i^*  \ketbra{\psi_i}{\psi_i}$ has rank at most~$r$, achieves objective value at least~$c$, and is feasible for~(SDP-P).
\end{proof}

\subsection{Conclusion for matrix edge-weighted graphs} \label{sec:dual-wrapup}

Let $V_\infty$ be a countable set of nodes and let $G_\infty$ be a bounded-degree graph on~$V_\infty$ with matrix edge-weights from~$\C^{r \times r}$.
Let $A_\infty$ be the adjacency operator for $G_\infty$, acting on $\ell_2(V_\infty) \otimes \C^r$ and assumed self-adjoint; we may think of it as an infinite matrix with rows and columns indexed by~$V_\infty$, and with entries from~$\C^{r \times r}$.
As in \Cref{def:extension}, we write $\wt{A}_\infty$ for its ``extension''; this is a self-adjoint bounded operator on $\ell_2(V_\infty \times [r])$.
Conversely, given any solution $\hat{\rho}$ for (SDP-P) from \Cref{sec:sdp-duality}, we may ``unextend'' it and think of it as an infinite matrix~$\rho$ with rows and columns indexed by~$V_\infty$, and entries from $\C^{r \times r}$.

We apply the the SDP duality theory from \Cref{sec:sdp-duality,sec:dual-finite}, with ``$V$'' being $V_\infty \times [r]$, ``$a$'' being $\wt{A}_\infty$, and ``$V_j$'' being $V_\infty \times \{j\}$ for all $1 \leq j \leq r$. 

The conclusion is that for any $\eps > 0$, there exists:
\begin{itemize}
    \item  $\hat{\oldu} \in \R^r$ with $\avg_j(\hat{\oldu}_j) = 0$;
    \item a finite subset~$F \subset V_\infty$;
    \item a PSD matrix $\rho$ (the ``unextension'' of~$\tilde{\rho}$ of rank at most~$r$ from \Cref{prop:rank-r}) with rows/columns indexed by~$V_\infty$ and entries from $\C^{r \times r}$, supported on the rows/columns~$F$,
        with
        \[
            \tr(\rho)_{jj} = \tfrac1r, \quad j = 1 \dots r;
        \]
\end{itemize}
such that for
\[
    \hat{A} = \wt{A}_\infty + \Id_{V_\infty} \otimes \diag(\hat{\oldu}),
\]
we have
\begin{equation}    \label{eqn:all}
    s^* \coloneqq \lambda_{\max}(\hat{A}) = \GenEigDual(A_\infty) = \GenEig(A_\infty) \geq \la \rho, A_\infty \ra = \la \rho_F, A_F \ra \geq s^* - \eps,
\end{equation}
where $\rho_F, A_F$ denote $\rho, A_\infty$ (respectively) restricted to the rows/columns~$F$.

\section{The SDP value of random matrix polynomial lifts}\label{sec:theproof}
In this section we prove our main \Cref{thm:ourmain}.
To that end, let $p$ be any self-adjoint matrix polynomial over indeterminates $Y_1, \ldots, Y_d, Z_1, \ldots , Z^*_e$ with $r \times r$ coefficients.
Let $A_\infty = p(\calL_\infty)$ denote the adjacency operator of the infinite lift $\mathfrak{X}_p$, and write $s^* = \GenEig(A_\infty) = \GenEigDual(A_\infty)$ as in
\Cref{eqn:all}.
Fix any $\eps, \beta > 0$, and let $\bA_n = p(\bcalL)$ denote the adjacency matrix of a corresponding $n$-lift, formed from $\bcalL = (\bM_1, \dots, \bM_d, \bP_1, \dots, \bP_e)$.
Our goal is to show that except with probability at most~$\beta$ (assuming $n$ is sufficiently large),
\[
    s^* - \eps \leq \GW(\bA_n) = \GWDual(\bA_n) \leq s^* + \eps.
\]

Given our setup, the upper bound follows easily from prior work, namely \Cref{thm:bc-with-edge-signing}.
Let $\hat{\oldu}$ and $\hat{A}$ be as in \Cref{sec:dual-wrapup}, and consider the matrix polynomial $p'$ defined by
\[
    p' = p + \diag(\hat{\oldu}) \bbone.
\]
Then on one hand, the $\infty$-lift of~$p'$ has adjacency operator precisely~$\hat{A}$; on the other hand,
\[
    \bA'_n \coloneqq p'(\bcalL) = \bA_n + \Id_{n\times n} \otimes \diag(\hat{\oldu}).
\]
Thus \Cref{thm:bc-with-edge-signing} tells us that except with probability~$\beta/2$ (provided~$n$ is large enough), the spectra
$\spec(\bA'_n)$ and $\spec(\hat{A})$ are at Hausdorff distance at most~$\eps$, from which it follows that
\[
    \lambda_{\max}(\bA'_n) \leq \lambda_{\max}(\hat{A}) + \eps = s^* + \eps.
\]
But this indeed proves $\GWDual(\bA_n) \leq s^* + \eps$, because $\hat{\oldu}$ has $\avg(\hat{\oldu}) = 0$ and hence is feasible for  $\GWDual(\bA_n)$.

It therefore remains to prove $\GW(\bA_n) \geq s^* - \eps$.

\subsection{A lower bound on the basic SDP value}

In this section we complete the proof of our main theorem by showing that $\GW(\bA_n) \geq s^* - \eps$ except with probability at most $o(1) = o_{n \to \infty}(1)$ (which is at most $\beta/2$ as needed, provided $n$ is large enough).

Let $F$, $\rho$, $\rho_F$, $A_F$ be as in \Cref{sec:dual-wrapup}, except with that section's ``$\eps$'' replaced by $\eps/2$, so that $\la \rho_F, A_F \ra \geq s^* - \eps/2$.
Adding finitely many vertices to~$F$ if necessary, we may assume that it consists of all reduced words over $Y_1, \dots, Y_d, Z_1, \dots, Z_e^*$ of length at most some finite~$f_0$.
We also make the following definition:
\begin{definition}[Cycle in a lift]
    Given an $n$-lift $\calL$, a \emph{cycle of length~$\ell > 0 $} is a pair $(i,w)$, where $i \in [n]$ and $w$ is a reduced word of length~$\ell$ such that $\calL^w \ket{i} = \pm \ket{i}$, and $\braket{i | \calL^{w'}| i} = 0$ (i.e., $\calL^{w'} \ket{i} \neq \pm \ket{i}$) for all proper prefixes~$w'$ of~$w$.
\end{definition}
We will employ the following basic random graph result, \cite[Lem.~23]{BC19}, stated in our language:
\begin{lemma}                                       \label{lem:23}
    For the random $n$-lift $\bcalL$, the expected number of cycles of length~$\ell$ is $O(\ell (d+2e-1)^\ell)$.
\end{lemma}
Applying this for all $\ell \leq f \coloneqq 2f_0 + \deg(p)$ and using Markov's inequality, we conclude:
\begin{corollary}                                       \label{cor:23}
    Except with probability at most $n^{-.99}$, the random $n$-lift $\bcalL$ has at most $O(n^{.99})$ cycles of length at most~$f$.
    In this case, we can exclude a set of ``bad'' vertices $B \subseteq [n]$ with $|B|/n \leq O(n^{-.01}) = o(1)$ so that:
    \[
        \forall i \not \in B, \quad \forall \text{ reduced words~$w$ with } 0 < |w| \leq 2f_0 + \deg(p), \qquad \braket{i |\bcalL^w|i} = 0.
    \]
\end{corollary}
We henceforth fix an outcome $\bcalL = \calL$ (and hence $\bA_n = A_n$) such that the conclusion of \Cref{cor:23} holds, accruing our $o(1)$ probability of failure.
Under this assumption, we will show $\GW(A_n) \geq s^* - \eps/2 - o(1)$, which is sufficient to complete the proof.\\

Our plan will be to first construct a ``provisional'' near-feasible PSD solution~$\sigma$ for $\GW(A_n)$, with rows/columns indexed by~$[n]$ and with $r \times r$ entries,  such that:
\begin{itemize}
    \item $|\la \sigma, A_n \ra - \la \rho_F, A_F \ra| \leq o(1)$, and hence $\la \sigma, A_n \ra \geq s^* - \eps/2 - o(1)$;
    \item for $i \not \in B$, the $r \times r$ matrix $\sigma_{ii}$ has diagonal entries~$\tfrac1{nr}$.
\end{itemize}
Then, we will show how to ``fix'' $\sigma$ to a some $\sigma'$ that is truly feasible for $\GW(A_n)$, while still having $\la \sigma', A_n \ra \geq \la \sigma, A_n \ra - o(1) \geq s^* - \eps/2 - o(1)$.

\subsubsection{Constructing a near-feasible solution}
\begin{definition}
    For each $i \in [n]$, we define a linear operator $\Phi_i : \C^F \to \C^n$ by
    \[
        \Phi_i = \sum_{v \in F} \calL^v \ketbra{i}{v},
    \]
    and also  $\wt{\Phi}_i = \Phi_i \otimes \Id_{r \times r}$.
    We furthermore define $\sigma_i$ to be the $n \times n$ matrix with $r \times r$ entries whose extension~$\wt{\sigma}_i$ is
    \[
        \wt{\sigma}_i = \wt{\Phi}_i  \cdot \wt{\rho}_F \cdot \wt{\Phi}_i^\transp.
    \]
\end{definition}
\begin{remark}
    $\wt{\sigma}_i$ is PSD, being the conjugation by $\wt{\Phi}_i$ of the PSD operator $\wt{\rho}_F$.
\end{remark}
\begin{proposition}                                     \label{prop:isomorphism}
    If $i \not \in B$, then $\wt{\Phi}_i^\transp \wt{A}_n \wt{\Phi}_i = \wt{A}_F$.
\end{proposition}
\begin{proof}
    Thinking of $\wt{\Phi}_i^\transp \wt{A}_n \wt{\Phi}_i$ as an $F \times F$ matrix of $r \times r$ matrices, it follows that its $(u,v)$ entry is given by
    \[
        \sum_{\text{term } a_w w \text{ in } p} \braket{i|\calL^{v^* w u}|i} a_w.
    \]
    On the other hand, the $(u,v)$ entry of $A_F$ is by definition
    \[
        \sum_{\text{term } a_w w \text{ in } p} 1[v^* w u = \bbi] a_w,
    \]
    where ``$v^* w u = \bbi$'' denotes that the reduced form of word $v^* w u$ is the empty word.
    We therefore have equality for all $u,v$ provided $\braket{i|\calL^{v^* w u}|i} = 0$ whenever $v^* w u \neq \emptyset$.
    But \Cref{cor:23} tells us this indeed holds for $i \not \in B$, because $|v^*wu| \leq 2f_0 + \deg(p)$.
\end{proof}
\begin{corollary}                                       \label{cor:equal-val}
    For $i \not \in B$ we have $\la \sigma_i, A_n \ra = \la \rho_F, A_F \ra$.
\end{corollary}
\begin{proof}
    When $i \not \in B$,
    \[
        \la \sigma_i, A_n \ra = \tr(\wt{\sigma}_i \wt{A}_n) = \tr(\wt{\Phi}_i  \wt{\rho}_F \wt{\Phi}_i^\transp \wt{A}_n) =
        \tr(\wt{\rho}_F \wt{\Phi}_i^\transp \wt{A}_n\wt{\Phi}_i) = \tr(\wt{\rho}_F \wt{A}_F) =  \la \rho_F, A_F \ra,
    \]
    where the last equality used \Cref{prop:isomorphism}
\end{proof}

We now define our ``provisional'' SDP solution~$\sigma$ via
\[
    \sigma = \avg_{i \in [n]} \{\sigma_i\};
\]
this is indeed PSD, being the average of PSD operators.
Using \Cref{cor:equal-val}, $|B|/n = o(1)$, and the fact that $|\la \sigma_i, A_n \ra| \leq O(1)$ for every~$i$ (since $\sigma_i$ only has $O(1)$ nonzero entries, each bounded in magnitude by~$O(1)$), we conclude:
\begin{proposition}                                     \label{prop:val}
    $|\la \sigma, A_n\ra - \la \rho_F, A_F \ra| \leq o(1)$, and hence $\la \sigma, A_n\ra  \geq s^* - \eps/2 - o(1)$.
\end{proposition}

Now similar to \Cref{prop:iso2} we have the following:
\begin{proposition}                                     \label{prop:iso2}
    If $j \not \in B$, then the $(j,j)$ entry of $\sigma$ is $\tfrac1n \tr(\rho_F)$ (and hence is an $r \times r$ matrix with diagonal entries equal to $\tfrac{1}{nr}$).
\end{proposition}
\begin{proof}
    By definition, the $(j,j)$ entry of $\sigma$ is
    \begin{align*}
        {}&= \avg_{i \in [n]} \sum_{u,v \in F} \braket{j|\calL^u|i} \braket{u|\rho_F|v} \braket{i | \calL^{v^*}|j}\\
        &= \tfrac1n \sum_{u,v \in F} \braket{u|\rho_F|v} \cdot \sum_{i \in [n]} \braket{j|\calL^u|i}  \braket{i | \calL^{v^*}|j} \\
        &= \tfrac1n \sum_{u,v \in F} \braket{u|\rho_F|v} \cdot  \braket{j|\calL^{uv^*}|j}
        \tag{since $\sum_i \ketbra{i}{i} = \Id$}
    \end{align*}
    where we are writing $\braket{u|\rho_F|v}$ for the $r \times r$ matrix at the $(u,v)$ entry of~$\rho_F$.
    Now when $j \not \in B$, we have that $\braket{j|\calL^{vw^*}|j} = 1[vw^* = \emptyset]$ by \Cref{cor:23}, since $|vw^*| \leq 2f_0$.
    Thus all summands above drop out, except for the ones with $u = v$; this indeed gives $\tfrac1n \tr(\rho_F)$.
\end{proof}

\subsubsection{Fixing $\sigma$}
    Finally, we slightly fix $\sigma$ to make it truly feasible for $\GW(A)$.
    Let $\sigma'$ be the $n \times n$ matrix, with entries from $\C^{r \times r}$, defined as follows:
    \[
        \sigma'_{ij} = \begin{cases}
                                    \sigma_{ij} & \text{if $i,j \not \in B$,} \\
                                    \tfrac{1}{nr} \Id_{r \times r} & \text{if $i = j \in B$,} \\
                                    0 & \text{else.}
                             \end{cases}
    \]
    This $\sigma'$ is easily seen to be PSD, being a principal submatrix of the PSD matrix~$\sigma$, direct-summed with the PSD matrix $\tfrac{1}{nr}\Id_{r \times r}$.
    As well, the  $nr \times nr$ extension matrix $\wt{\sigma}'$ has all diagonal entries equal to~$\tfrac{1}{nr}$, by \Cref{prop:iso2}.
    Thus $\wt{\sigma}'$ is feasible for $\GW(A_n) = \GW(\wt{A}_n)$, and it remains for us to show that
     \begin{equation}  \label[ineq]{ineq:finish}
        \la \sigma, A_n \ra - \la \sigma', A_n \ra \leq o(1);
    \end{equation}
    this will imply $\la \sigma', A_n \ra \geq s^* - \eps/2 - o(1)$ by \Cref{prop:val}, and hence $\GW(A_n) \geq s^* - \eps$ (for sufficiently large~$n$), as desired.

    We have
    \[
       \la \sigma, A_n \ra - \la \sigma', A_n \ra = \la \sigma - \sigma', A_n \ra \leq \|\wt{\sigma} - \wt{\sigma}'\|_1 \|\wt{A}_n\|_\infty.
    \]
    Next,
    \[
        \|\wt{A}_n\|_\infty = \|p(\calL)\|_\infty = \|\sum_{w} \calL^w \otimes a_w\|_\infty \leq \sum_{w} \|\calL^w\|_\infty \cdot \|a_w\|_\infty = \sum_w \|a_w\|_\infty \leq O(1).
    \]
    Here the final equality is because each $\calL^w$ is a signed permutation matrix (hence has $\|\calL^w\|_\infty = 1$), and the final inequality is because~$p$ has only constantly many coefficients, of constant size.
    Thus to establish \Cref{ineq:finish}, it remains to show $\|\wt{\sigma} - \wt{\sigma}'\|_1 \leq o(1)$.

    Define the orthogonal projection matrices $\Pi_B = \sum_{i \in B} \ketbra{i}{i} \otimes \Id_{r \times r}$ and similarly $\Pi_{\ol{B}}$, where $\ol{B} = [n] \setminus B$.
    Observe that
    \[
        \sigma' = \tfrac{1}{nr} \Pi_B + \Pi_{\ol{B}} \sigma \Pi_{\ol{B}},
    \]
    and thus
    \[
        \|\wt{\sigma} - \wt{\sigma}'\|_1 = \|\wt{\sigma} - \wt{\Pi}_{\ol{B}} \wt{\sigma} \wt{\Pi}_{\ol{B}} - \tfrac{1}{nr} \wt{\Pi}_B\|_1 \leq \|\wt{\sigma} - \wt{\Pi}_{\ol{B}} \wt{\sigma} \wt{\Pi}_{\ol{B}}\|_1 + \tfrac{1}{nr} \|\wt{\Pi}_B\|_1.
    \]
    But $\tfrac{1}{nr} \|\wt{\Pi}_B\|_1 = \tfrac{|B|}{n} = o(1)$, so it remains to show
    \[
        \|\wt{\sigma} - \wt{\Pi}_{\ol{B}} \wt{\sigma} \wt{\Pi}_{\ol{B}}\|_1 \leq o(1).
    \]
    Note that $\wt{\sigma} \in \C^{nr \times nr}$ is nearly a density matrix: it is PSD, and all but a $1 - o(1)$ fraction of its diagonal entries are~$\tfrac{1}{nr}$, with the remaining ones being bounded in magnitude by $O(\tfrac{1}{n})$.
    Thus $\tr(\wt{\sigma}) = 1 \pm o(1)$, and we can therefore scale $\wt{\sigma}$ by a $1\pm o(1)$ factor to produce a true density matrix~$\hat{\sigma}$.
    Clearly it now suffices to show
    \[
        \|\hat{\sigma} - \wt{\Pi}_{\ol{B}} \hat{\sigma} \wt{\Pi}_{\ol{B}}\|_1 \leq o(1).
    \]
    But this follows from Winter's Gentle Measurement Lemma~\cite[Lem.~9]{Win99}, which bounds the  quantity on the left by $\sqrt{8\lambda}$, where $\lambda = 1 - \tr(\hat{\sigma} \wt{\Pi}_{\ol{B}}) = o(1)$.
    This completes the proof.

\bibliographystyle{alpha}
\bibliography{arxiv}

\end{document}